\documentclass[preprint,a4paper,11pt]{elsarticle}
\usepackage{fullpage}

\usepackage{mathtools}

\usepackage{graphicx}
\usepackage{amsmath,amssymb}
\usepackage{comment}
\usepackage{hyperref}
\usepackage{bm}

\usepackage{algorithm}
\usepackage[noend]{algpseudocode}

\usepackage{mleftright}

\usepackage{color}
\definecolor{darkred}{rgb}{0.7,0,0}

\newcommand{\Order}{\mathrm{O}}

\newtheorem{theorem}{Theorem}[section]
\newtheorem{lemma}[theorem]{Lemma}
\newdefinition{observation}[theorem]{Observation}
\newdefinition{claim}[theorem]{Claim}
\newproof{proof}{Proof}

\newcommand{\localqed}{\hfill$\diamond$}

\newcommand{\fig}{fig}
\newcommand{\figref}[1]{\figurename~\ref{#1}}

\begin{document}

\journal{arXiv} %TODO: update this later

\title{Longest Common Subsequence in Sublinear Space\tnoteref{t1}} % \tnoteref{t1,t2}
\tnotetext[t1]{Partially supported by JSPS KAKENHI grant numbers 
JP18H04091, % Yota (Uehara A)
JP18K11153, % Takashi (Horiyama C)
JP18K11168, % Yota (Otachi C)
JP18K11169. % Masashi and Yota (Kiyomi C)
This work was started at the 3rd Japan-Canada Joint Workshop on Computational Geometry and Graphs
held on January 7--9, 2018 in Vancouver, Canada.
The authors thank the other participants of the workshop for the fruitful discussions.
}
%\tnotetext[t2]{others}

\author[1]{Masashi Kiyomi\corref{cor1}%
%  \fnref{fn1}
}
\ead{masashi@yokohama-cu.ac.jp}

\author[2]{Takashi Horiyama%
%  \fnref{fn2}
}
\ead{horiyama@ist.hokudai.ac.jp}

\author[3]{Yota Otachi%
%  \fnref{fn3}
}
\ead{otachi@nagoya-u.jp}

\cortext[cor1]{Corresponding author}
%\fntext[fn1]{This is the first author footnote.}
%\fntext[fn2]{This is the first author footnote.}
%\fntext[fn3]{This is the first author footnote.}

\address[1]{School of Data Science, Yokohama City University}
\address[2]{Faculty of Information Science and Technology, Hokkaido University}
\address[3]{Department of Mathematical Informatics, Nagoya University}

\begin{abstract}
We present the first $\mathrm{o}(n)$-space polynomial-time algorithm for
computing the length of a longest common subsequence.
Given two strings of length $n$, the algorithm runs in $\Order(n^{3})$ time with
$\Order\mleft(\frac{n \log^{1.5} n}{2^{\sqrt{\log n}}}\mright)$ bits of space.
\end{abstract}

\begin{keyword}
longest common subsequence, space-efficient algorithm
\end{keyword}

\maketitle

\section{Introduction}

A \emph{subsequence} of a string is a string that can be obtained from the original string by removing some elements.
If two strings $X$ and $Y$ contain a string $Z$ as their subsequences, then $Z$ is a \emph{common subsequence} of $X$ and $Y$.
For example, ``tokyo'' and ``kyoto'' have ``to'' and ``oo'' as common subsequences but not ``too.''
The problem of finding (the length of) a \emph{longest common subsequence} is a classic problem in computer science.
The applications of the problem range over many fields
(see \cite{Hirschberg83,PatersonD94,Apostolico97,BergrothHR00} and the references therein).
Given two strings of length $n$,
the problem can be solved in $\Order(n^2)$ time with $\Order(n \log n)$ bits of space using
a dynamic programming approach~\cite{WagnerF74,Hirschberg75}.
On the other hand, it is known that under the strong exponential time hypothesis,
there is no $\Order(n^{2-\varepsilon})$-time algorithm for any $\varepsilon > 0$~\cite{BackursI18,AbboudBW15,BringmannK15,AbboudHWW16}.
Given this lower bound, subquadratic-time approximation algorithms have been 
proposed~(see \cite{HajiaghayiSSS19,RubinsteinS20} and the references therein).
Recently, Cheng et al.~\cite{Cheng20arxiv} presented
an approximation algorithm of factor $1 - \mathrm{o}(1)$
that runs in polynomial time with polylogarithmic space.

In this paper, we seek for \emph{exact} (i.e., non-approximation) polynomial-time algorithms with small space complexity.
To the best of authors' knowledge, there was no known polynomial-time algorithm that runs with $\mathrm{o}(n)$ bits of space.
A natural (but probably quite challenging) question in this direction would be whether there is a polynomial-time algorithm that runs with 
\emph{truly} sublinear space of $\Order(n^{1-\varepsilon})$ bits for some $\varepsilon > 0$~\cite{Cheng20arxiv}.
We make a step toward this goal by giving a polynomial-time algorithm
that runs with \emph{slightly} sublinear space.
More precisely, the result of this paper is as follows.
\begin{theorem}
\label{thm:main}
Given two strings of length $n$, the length of a longest common subsequence
can be computed in $\Order(n^3)$ time with $\Order\mleft(\frac{n\log^{1.5} n}{2^{\sqrt{\log n}}}\mright)$ bits of space.
\end{theorem}

\section{The algorithm}
We use the standard computational model in the literature of space-efficient algorithms.
That is, we take the RAM model with the following restrictions:
\begin{itemize}
  \setlength{\itemsep}{0pt}
  \item the input is in a \emph{read-only} memory;
  \item the output must be produced on a \emph{write-only} memory;
  \item an additional memory that is \emph{readable and writable} can be used.
\end{itemize}
We measure space consumption in the number of bits used within the additional memory.
Throughout the paper, we fix the base of logarithms to $2$. That is, $\log x$ means $\log_{2} x$.

For nonnegative integers $m$ and $n$,
we denote by $\Gamma_{m,n}$ the directed acyclic graph such that 
\begin{align*}
  V(\Gamma_{m,n}) ={}& \{v_{i,j} \mid 0 \le i \le m, \, 0 \le j \le n\},
  \\
  E(\Gamma_{m,n}) 
  ={}& 
  \{(v_{i,j}, v_{i+1,j}) \mid 0 \le i \le m-1, \, 0 \le j \le n\} \cup {}
  \\
  &\{(v_{i,j}, v_{i,j+1}) \mid 0 \le i \le m, \, 0 \le j \le n-1\} \cup {}
  \\
  &\{(v_{i,j}, v_{i+1,j+1}) \mid 0 \le i \le m-1, \, 0 \le j \le n-1\}.
\end{align*}
Namely, $\Gamma_{m,n}$ is the $(m+1) \times (n+1)$ grid with the main diagonal edge in each square,
where each edge is oriented from left to right and from top to bottom. See \figref{fig:gamma}.

\begin{figure}[htb]
  \centering
  \includegraphics[scale=1.2]{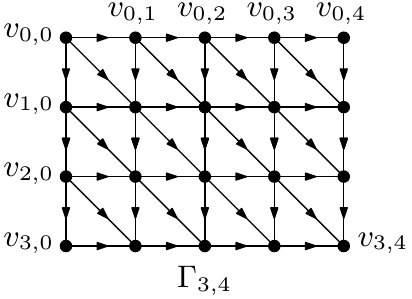}
  \caption{The directed acyclic graph $\Gamma_{m,n}$ with $m = 3$ and $n = 4$.}
  \label{fig:gamma}
\end{figure}

Let $S = s_{1} s_{2} \cdots s_{m}$ and $T = t_{1} t_{2} \cdots t_{n}$ be strings of length $m$ and $n$, respectively.
By $\Gamma(S,T)$, we denote the graph $\Gamma_{m,n}$ with the edge weights defined as follows:
the horizontal and vertical edges have weight $0$;
the diagonal edge $(v_{i,j}, v_{i+1,j+1})$ has weight $1$ if $s_{i+1} = t_{j+1}$; otherwise it has weight $0$.
It is easy to see that
every $v_{0,0}$--$v_{|S|, |T|}$ path is monotone in both $x$ and $y$ directions
and the positive weight edges in the path represent a common subsequence of $S$ and $T$.
Moreover, we can show the following fact, which is used in most of the existing algorithms.
\begin{observation}[Folklore]
The length of a longest common subsequence of $S$ and $T$ is equal to 
the longest path length from $v_{0,0}$ to $v_{|S|, |T|}$ in $\Gamma(S,T)$.
\end{observation}
As mentioned in the introduction, 
the length of a longest common subsequence of $S$ and $T$,
or equivalently the longest path length from $v_{0,0}$ to $v_{|S|, |T|}$ in $\Gamma(S,T)$,
can be computed in $\Order(mn)$ time with $\Order(m \log n)$ bits of space, where $m < n$.
We describe the idea of this algorithm in a slightly generalized setting.

We denote by $\Gamma_{m,n}^{w}$
the graph $\Gamma_{m,n}$ with an edge weighting $w$, where each edge weight can be represented in $\Order(\log n)$ bits.
Let $\lambda_{w}(v_{i,j}, v_{i',j'})$ be the longest path length from $v_{i,j}$ to $v_{i',j'}$ in $\Gamma_{m,n}^{w}$.
Then $\lambda_{w}$ can be expressed in a recursive way as follows:
\[
  \lambda_{w}(v_{0,0}, v_{i,j})
  =
  \begin{cases}
    0 & i = j = 0, \\
    \lambda_{w}(v_{0,0}, v_{0,j-1}) + w(v_{0,j-1}, v_{0,j})& i = 0,\ j \ge 1, \\
    \lambda_{w}(v_{0,0}, v_{i-1,0}) + w(v_{i-1,0}, v_{i,0}) & i \ge 1, \ j = 0, \\
    \max\left\{\begin{array}{l}
	\lambda_{w}(v_{0,0}, v_{i-1,j})   + w(v_{i-1,j}, v_{i,j}), \\
	\lambda_{w}(v_{0,0}, v_{i,j-1})   + w(v_{i,j-1}, v_{i,j}), \\
	\lambda_{w}(v_{0,0}, v_{i-1,j-1}) + w(v_{i-1,j-1}, v_{i,j}) \\
    \end{array}\right\} & i \ge 1, \ j \ge 1.
  \end{cases}
\]
Clearly, we can compute $\lambda_{w}(v_{0,0}, v_{m,n})$ in $\Order(mn)$ time.
While a naive implementation takes $\Order(mn \log n)$ bits of space,
by a simple trick of computing the entries for $v_{i,j}$ in an increasing order of $j$
and storing only two recent rows (that is, the rows $j-1$ and $j$),
we can reduce the space consumption to $\Order(m \log n)$ bits.
We call this recursive method the \emph{standard algorithm}.

We now explain the high-level idea of our algorithm.
Observe that the problem of deciding whether $\lambda_{w}(v_{0,0}, v_{m,n}) \ge \ell$ is in NL
since we can nondeterministically find the next vertex in a longest path in logspace
and we can forget the visited vertices as the graph is acyclic.
Hence, Savitch's theorem~\cite{Savitch70} implies that the problem can be solved deterministically with $\Order(\log^{2} n)$ bits of space
but in quasipolynomial ($n^{\Order(\log n)}$) time.
Such an algorithm recursively find the center of a longest path, and thus has a search tree of maximum degree $n$ and depth $\log n$.
Our algorithm can be seen as a modification of this algorithm.
Instead of guessing a single vertex included in a longest path, we guess a set of vertices that a longest path intersects.
The search tree of our algorithm has maximum degree $\Order(2^{\sqrt{\log n}})$ and depth $\Order(\sqrt{\log n})$.
This gives us a sublinear-space polynomial-time algorithm.

We now prove the main lemma based on the idea described above, which immediately implies Theorem~\ref{thm:main}.
\begin{lemma}
\label{lem:main}
For an edge weighting $w \colon E(\Gamma_{m,n}) \to \mathbb{Z} \cup \{-\infty\}$ given as a constant time oracle, 
where each edge weight can be represented in $\Order(\log n)$ bits,
the longest path length from $v_{0,0}$ to $v_{m,n}$ in $\Gamma_{m,n}^{w}$ can be computed in $\Order(n^{3})$ time
using $\Order\mleft(\frac{n \log^{1.5} n}{2^{\sqrt{\log n}}}\mright)$ bits of space.
\end{lemma}
\begin{proof}
Without loss of generality assume that $m \le n$.
We also assume that $m$ is a power of~$2$ by 
adding, if necessary, some dummy columns and rows and set the weights of the new edges to~$-\infty$.
The numbers of columns and rows are doubled in the worst case.
Note that we do not have to construct the dummy columns and rows explicitly.
We just need to remember the original $m$ and $n$ with $\Order(\log n)$ bits of additional space consumption.

Let $B = \lceil (n+1)/2^{\sqrt{\log n}} \rceil$.
We compute the length in a recursive way.
In the recursive algorithm described below,
$m$ and $n$ may get smaller in each recursive call, while we keep $B$ the same.
The algorithm solves a slightly generalized problem,
where we compute the longest path lengths from $v_{0,0}$ 
to the vertices in $T \coloneqq \{v_{m,j} \mid (\lceil (n+1)/B\rceil -1) B \le j \le n\}$.

\paragraph{The algorithm}
Our algorithm solves the problem in a recursive way.
(See Algorithm~\ref{alg:outline} for the outline.)
If $\min\{m,n\} \le 2B$, then we solve the problem with the standard algorithm.

Assume that $m, n > 2B$. 
Let $V_{h} = \{v_{m/2, j} \mid h B \le j \le \min\{(h+1) B -1, n\}\}$ for $0 \le h \le\lceil (n+1)/B \rceil -1$.
Observe that for every $u \in T$, each $v_{0,0}$--$u$ path intersects $V_{h}$ for some $h$. 
(Note that such $h$ is not necessarily unique.)
Such a path first goes through the upper-left part % of the grid 
induced by $\{v_{i,j} \mid 0 \le i \le m/2, \, 0 \le j \le \min\{(h+1)B - 1, n\}\}$, 
reaches $V_{h}$,
then goes through the bottom-right part % of the grid, 
induced by $\{v_{i,j} \mid m/2 \le i \le m, \, hB \le j \le n\}$, 
and finally reaches $u$.
See \figref{fig:recursion}.

\begin{figure}[htb]
  \centering
  \includegraphics[width=.99\textwidth]{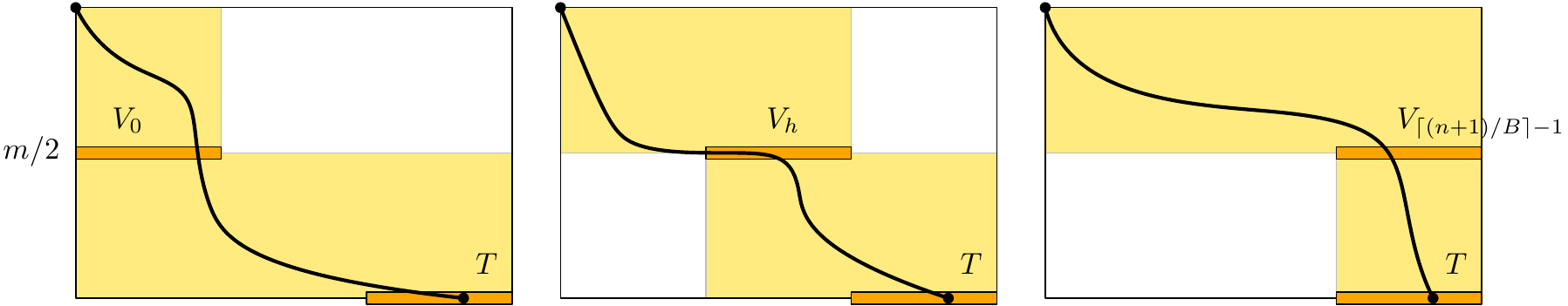}
  \caption{Every path intersecting $V_{h}$ goes through upper-left and bottom-right regions of the grid.}
  \label{fig:recursion}
\end{figure}

Based on the observation above, 
we divide the problem into $\lceil (n+1)/B \rceil$ pairs of subproblems as follows.
Let $h \in \{0,\dots,\lceil (n+1)/B \rceil - 1\}$.
We first find the longest path lengths from $v_{0,0}$ to the vertices of $V_{h}$ in $\Gamma_{m,n}^{w}$.
We find the lengths by recursively applying the algorithm to $\Gamma_{m/2,\min\{(h+1) B -1, n\}}^{w'}$,
where $w'$ is obtained from $w$ by restricting it to $\Gamma_{m/2,\min\{(h+1) B -1, n\}}$.
The recursive call returns the longest path lengths $\bm{\ell}'$ from $v_{0,0}$ to $V_{h}$,
where $\bm{\ell}'$ is a $|V_{h}|$-dimensional vector indexed by the second coordinates $hB, hB+1,\dots,\min\{(h+1)B-1,n\}$ of the vertices.

Let $w''$ be the edge weighting of $\Gamma_{m/2, n-hB}$ obtained from $w$ and $\bm{\ell}'$ as follows:
\begin{align*}
  w''(v_{0,j}, v_{0,j+1}) &= \bm{\ell}'(j+hB+1) - \bm{\ell}'(j+hB) &&\text{for } 0 \le j \le |V_{h}|-2,
  \\
  w''(v_{i,j}, v_{i',j'}) &= w(v_{i+m/2, j+hB}, v_{i'+m/2, j'+hB}) &&\text{otherwise}.
\end{align*}
Namely, $w''$ is obtained from $w$ by first restricting it into its bottom-right part starting at $V_{h}$,
and then changing the weight of the edges in $V_{h}$ with respect to the longest path lengths from $v_{0,0}$.
We now recursively apply the algorithm to $\Gamma_{m/2, n-hB}^{w''}$ for computing $\bm{\ell}''$,
the longest path lengths from $V_{h}$ to $T$ under $w''$.

Now we set $\bm{\ell}(j) = \max_{0 \le h \le \lceil (n+1)/B \rceil -1} \bm{\ell}'(hB) + \bm{\ell}''(j - hB)$
for $(\lceil (n+1)/B \rceil -1)B \le j \le n$.
The correctness of this final step follows from the claim below.
\begin{claim}
\label{clm:correctness-of-the-recursion}
The maximum length of a path $P$ from $v_{0,0}$ to $v_{m,j} \in T$ passing through $V_{h}$
is $\bm{\ell}'(hB) + \bm{\ell}''(j-hB)$.
\end{claim}
\begin{proof}
[of Claim~\ref{clm:correctness-of-the-recursion}]
Let $k$ be the maximum index such that $v_{m/2, k} \in V_{h}$ is included in $P$.
Then, the length of $P$ is $\lambda_{w}(v_{0,0},v_{m/2,k}) + \lambda_{w}(v_{m/2,k}, v_{m,j})$.

Let $Q$ be the subpath of $P$ from $v_{m/2,k}$ to $v_{m,j}$. 
The length of $Q$ is $\lambda_{w}(v_{m/2,k}, v_{m,j})$.
Let $P''$ be the $v_{0,0}$--$v_{m/2,j-hB}$ path in $\Gamma_{m/2,n-hB}^{w''}$
that first takes the unique $v_{0,0}$--$v_{0,k-hB}$ path of length 
$\sum_{0 \le p \le k-hB-1}(\bm{\ell}'(p+hB+1) - \bm{\ell}'(p+hB)) = \bm{\ell}'(k) - \bm{\ell}'(hB)$
and then follows $Q$ by shifting each vertex coordinate by $-(m/2, hB)$.
Since $\bm{\ell}'(k) = \lambda_{w}(v_{0,0},v_{m/2,k})$,
the length of $P''$ is 
$\lambda_{w}(v_{0,0},v_{m/2,k})+ \lambda_{w}(v_{m/2,k}, v_{m,j}) - \bm{\ell}'(hB)$.
Since $P''$ is a $v_{0,0}$--$v_{m/2,j-hB}$ path in $\Gamma_{m/2,n-hB}^{w''}$,
it holds that 
\begin{align}
  \bm{\ell}''(j-hB) \ge \lambda_{w}(v_{0,0},v_{m/2,k})+ \lambda_{w}(v_{m/2,k}, v_{m,j}) - \bm{\ell}'(hB). 
  \label{eq:correctness-of-the-recursion-1}
\end{align}
 
Let $Q''$ be a $v_{0,0}$--$v_{m/2,j-hB}$ path of length $\bm{\ell}''(j-hB)$ in $\Gamma_{m/2,n-hB}^{w''}$.
Let $q$ be the maximum index such that $q \le |V_{h}| - 1$ and $v_{0,q}$ is included in $Q''$.
The unique $v_{0,0}$--$v_{0,q}$ path in $\Gamma_{m/2,n-hB}^{w''}$ has length 
$\sum_{0 \le p \le q-1}(\bm{\ell}'(p+hB+1) - \bm{\ell}'(p+hB)) = \bm{\ell}'(q+hB) - \bm{\ell}'(hB)$.
From the construction of $w''$, the rest of $P''$ starting at $v_{0,q}$ has length
$\lambda_{w''}(v_{0,q}, v_{m/2,j-hB}) = \lambda_{w}(v_{m/2,q+hB}, v_{m,j})$.
Since $\bm{\ell}'(q+hB) = \lambda_{w}(v_{0,0}, v_{m/2,q+hB})$,
\begin{align}
  \bm{\ell}''(j-hB) 
  &=
  \lambda_{w}(v_{0,0},v_{m/2,q+hB}) + \lambda_{w}(v_{m/2,q+hB}, v_{m,j}) - \bm{\ell}'(hB)  \nonumber
  \\
  &\le 
  \lambda_{w}(v_{0,0},v_{m/2,k}) + \lambda_{w}(v_{m/2,k}, v_{m,j}) - \bm{\ell}'(hB). 
  \label{eq:correctness-of-the-recursion-2}
\end{align}
Equations~\eqref{eq:correctness-of-the-recursion-1} and \eqref{eq:correctness-of-the-recursion-2}
imply that $\bm{\ell}''(j-hB) + \bm{\ell}'(hB) = \lambda_{w}(v_{0,0},v_{m/2,k}) + \lambda_{w}(v_{m/2,k}, v_{m,j})$.
\localqed
\end{proof}

\begin{algorithm}
  \caption{The algorithm given in the proof of Lemma~\ref{lem:main}.}
  \label{alg:outline}
  \begin{algorithmic}[1]
    \Procedure{LongestPathLength}{$\Gamma_{m,n}^{w}$, $B$}
    \If{$\min\{m,n\} \le 2B$}
      \State Use the standard algorithm and return the longest path lengths.
    \Else
      \State $\bm{\ell} \coloneqq \bm{0}$ \Comment{$\lambda_{w}(v_{0,0}, v_{m,j})$ for $(\lceil (n+1)/B\rceil -1) B \le j \le n$}
      \For{$h \in \{0,\dots,\lceil (n+1)/B\rceil-1\}$}
        \State $\bm{\ell}' \coloneqq {}$\Call{LongestPathLength}{$\Gamma_{m/2,\min\{(h+1) B -1, n\}}^{w'}$, $B$}
	\State Compute $w''$ from $w$ and $\bm{\ell}'$.
        \State $\bm{\ell}'' \coloneqq {}$\Call{LongestPathLength}{$\Gamma_{m/2,n-hB}^{w''}$, $B$}
	\For{$j \in \{(\lceil (n+1)/B\rceil -1) B, \dots, n\}$}
	  \State $\bm{\ell}(j) \coloneqq \max\{\bm{\ell}(j), \, \bm{\ell}'(hB) + \bm{\ell}''(j-hB)\}$
	\EndFor
      \EndFor
      \State \Return $\bm{\ell}$
    \EndIf
    \EndProcedure
  \end{algorithmic}
\end{algorithm}

\paragraph{Space consumption}
In the recursion tree of Algorithm~\ref{alg:outline}, each inner node
stores $\Order(B \log n)$ bits for $\ell$, $\ell'$, and $\ell''$.
The weight functions $w'$ and $w''$ are also stored at each inner node
and provided as constant time oracles for the children in the recursive tree.
Observe that $w'$ can be represented by $w$ with a constant number of indices,
and $w''$ can be represented by $w$ with a constant number of indices and $\bm{\ell}'$.
In total, an inner node stores $\Order(B \log n)$ bits of information.
Each leaf node executes the standard algorithm with $\min\{m,n\} \le 2B$, and thus only needs $\Order(B \log n)$ bits.
Since the depth of the recursion tree is $\lceil \log (m/(2B)) \rceil < \sqrt{\log n}$,
the total space consumption is $\Order(B \log n \cdot \sqrt{\log n}) = \Order(n \log^{1.5} n / 2^{\sqrt{\log n}})$.

\paragraph{Running time}
We estimate an upper bound $L(m)$ of the number of leaves in the recursion tree, where $m+1$ is the number of rows.
If $m \le 2B$, then $L(m) \le 1$.
Assume that $m > 2B$. Then
\[
  L(m) \le 2 \lceil (n+1)/B \rceil \cdot L(m/2) < (2 (2^{\sqrt{\log n}} + 1))^{\sqrt{\log n}}
\]
since the algorithm invokes at most $2\lceil (n+1)/B \rceil \le 2^{\sqrt{\log n}} + 1$ recursive calls with the parameter $m/2$
and the depth of recursion is at most $\lceil \log (m/(2B)) \rceil < \sqrt{\log n}$.
Since $(2^{\sqrt{\log n}} + 1)^{\sqrt{\log n}} \le e \cdot (2^{\sqrt{\log n}})^{\sqrt{\log n}} \le en$,
it holds that $L(m) \in \Order(2^{\sqrt{\log n}} \cdot n)$.
Since each inner node of the recursion tree has two or more children,
the number of all nodes in the recursion tree is also $\Order(2^{\sqrt{\log n}} \cdot n)$.
Each leaf node takes $\Order(B n)$ time,
and each inner node takes $\Order(n)$ time excluding the time spent by its children.
Therefore, the total running time is $\Order(B n \cdot 2^{\sqrt{\log n}} \cdot n) = \Order(n^{3})$.
\qed
\end{proof}

\section{Concluding Remarks}
We have presented an algorithm for computing the length of a longest common subsequence of two string of length $n$
in $\Order(n^{3})$ time using $\Order\mleft(\frac{n \log^{1.5} n}{2^{\sqrt{\log n}}}\mright)$ bits of space.
The challenge for finding a polynomial-time algorithm with 
$\Order(n^{1-\varepsilon})$ space for some constant $\varepsilon > 0$ remains unsettled.

Our algorithm in Lemma~\ref{lem:main} is designed for a slightly general problem on the edge weighted grid-like graph $\Gamma_{m,n}^{w}$.
The generality allows us to compute some other similarity measures of strings as well.
For example, the \emph{edit distance} (or the \emph{Levenshtein distance}) between two strings is the minimum number of operations
(insertions, deletions, or substitutions) required to make the strings the same.
We can formulate this problem as the shortest path problem on $\Gamma_{m,n}^{w}$
even in a general setting where each operation has different cost possibly depending on the symbols involved.
Since the shortest path problem on $\Gamma_{m,n}^{w}$ is equivalent to 
the longest path problem on $\Gamma_{m,n}^{-w}$,
Lemma~\ref{lem:main} implies that we can compute the (general) edit distance in the same time and space complexity.

\bibliographystyle{plainurl}
\bibliography{lcs}

\end{document}